\documentclass[letterpaper,10pt,conference]{ieeeconf}

\usepackage{cite}
\usepackage{graphicx}
\usepackage{tikz}
\usepackage{subcaption}
\usepackage[cmex10]{amsmath}
\usepackage{algpseudocode}
\usepackage{algorithmicx} 
\usepackage{array}
\usepackage{multirow}
\usepackage{dblfloatfix}

\usepackage{epsfig}
\usepackage{amsfonts,amssymb,multirow,bigstrut,booktabs,ctable,latexsym}

\usepackage[mathscr]{eucal}
\usepackage{verbatim}

\usepackage{amsthm}
\usepackage{algorithm}

\usepackage{stmaryrd}

\IEEEoverridecommandlockouts  
\begin{document}

\newtheorem{definition}{Definition}
\newtheorem{lemma}{Lemma}
\newtheorem{corollary}{Corollary}
\newtheorem{theorem}{Theorem}
\newtheorem{example}{Example}
\newtheorem{proposition}{Proposition}
\newtheorem{remark}{Remark}
\newtheorem{assumption}{Assumption}
\newtheorem{corrolary}{Corrolary}
\newtheorem{property}{Property}
\newtheorem{ex}{EX}
\newtheorem{problem}{Problem}
\newcommand{\argmin}{\arg\!\min}
\newcommand{\argmax}{\arg\!\max}
\newcommand{\st}{\text{s.t.}}
\newcommand \dd[1]  { \,\textrm d{#1}  }

\makeatother

\title{\Large\bf An Analytical Framework for Control Synthesis of Cyber-Physical Systems with Safety Guarantee}

\author{Luyao Niu$^{1*}$, Abdullah Al Maruf$^{2*}$, Andrew Clark$^1$, J. Sukarno Mertoguno$^3$, and Radha Poovendran$^2$ %
\thanks{*Authors contributed equally to this work.}
\thanks{$^1$Luyao Niu and Andrew Clark are with the Department of Electrical and Computer Engineering, Worcester Polytechnic Institute, Worcester, MA 01609
{\tt\small \{lniu,aclark\}@wpi.edu}}
\thanks{$^2$Abdullah Al Maruf and Radha Poovendran are with the Network Security Lab, Department of Electrical and Computer Engineering,
University of Washington, Seattle, WA 98195-2500
        {\tt\small \{maruf3e,rp3\}@uw.edu}}%
\thanks{$^3$J. Sukarno Mertoguno is with Information and Cyber Sciences Directorate, Georgia Tech Research Institute, Atlanta, GA 30332
{\tt\small \{karno\}@gatech.edu}}
 }
\thispagestyle{empty}
\pagestyle{empty}

\maketitle

\begin{abstract}
Cyber-physical systems (CPS) are required to operate safely under fault and malicious attacks. The simplex architecture and the recently proposed cyber resilient architectures, e.g., Byzantine fault tolerant++ (BFT++),  provide safety for CPS under faults and malicious cyber attacks, respectively. However, these existing architectures make use of different timing parameters and implementations to provide safety, and are seemingly unrelated. In this paper, we propose an analytical framework to represent the simplex, BFT++ and other practical cyber resilient architectures (CRAs). We construct a hybrid system that models CPS adopting any of these architectures. We derive sufficient conditions via our proposed framework under which a control policy is guaranteed to be safe. We present an algorithm to synthesize the control policy. We validate the proposed framework using a case study on lateral control of a Boeing 747, and demonstrate that our proposed approach ensures safety of the system.


\end{abstract}
\section{Introduction}\label{sec:intro}
Cyber-physical systems (CPS) are subject to random failures and malicious cyber attacks, which have been reported in applications such as transportation \cite{Jeep} and power system \cite{case2016analysis}. Failures and attacks can potentially cause safety violation of the physical components, which leads to severe harm to the plants and humans. 


Fault tolerant control schemes \cite{zhang2008bibliographical,sharifi2010fault,xu2020distributed} and architectures such as simplex \cite{sha2001using,bak2009system,mohan2013s3a} have been proposed to address random failures. These approaches are effective in CPS when some components are verified to be fault-free. This requirement, however, may not be viable for all CPS, especially those subject to malicious attacks.


A malicious adversary can exploit the vulnerabilities in the cyber subsystem and intrude into CPS. The adversary can then cause common mode failures across different components, rendering fault-tolerant schemes designed for random failures inadequate. 
A seminal work recently proposed a cyber resilient architecture, named Byzantine fault tolerant++ (BFT++) \cite{mertoguno2019physics}, for CPS under malicious cyber attacks. 
BFT++, which is applied to CPS with redundant controllers, uses one of the controllers as backup. The other controllers are engineered to crash upon malicious attack, which triggers automatic and fast controller recovery using the backup. If the controllers are restored in time, then the system can guarantee safety.
This architecture exploits the fact that the cyber subsystem operates on a shorter timescale than the physical subsystem, which has an inherent natural resilience from the physical dynamics against limited cyber disruptions.


Following BFT++, several other unpublished yet effective approaches have appeared with different implementations of recovery and backup. 
In parallel, alternative approaches are proposed in \cite{arroyo2019yolo,arroyo2017fired} for CPS without redundancy. The controller in these approaches is programmed to restart proactively or periodically to recover the system from malicious attack. The physical subsystem can then maintain safety by utilizing its natural resilience and tuning the controller availability.


The aforementioned architectures \cite{mertoguno2019physics,arroyo2019yolo,arroyo2017fired}, which we collectively refer as cyber resilient architectures (CRAs) have found successful applications in different CPS.
While the CRAs can independently provide safety guarantees, the analyses undertaken are distinct and specific to the systems or architectures. Hence these analyses may not be readily extended from one CPS to another.  
Therefore a common framework which allows a general method of analysis for these seemingly unrelated yet novel architectures is of key interest. Such a framework will also enable comparison among different architectures under a common baseline. Currently, such an analytical framework does not exist.

In this paper, we propose a common framework that models the simplex architecture and the CRAs. We then present a control policy synthesis with safety guarantee using our proposed framework, which applies to any of these architectures. 
We make the following specific contributions:
\begin{itemize}
    \item We construct a hybrid system to model CPS implementing the simplex architecture and CRAs. We propose a common framework that captures these architectures.
    \item We derive the sufficient conditions for a control policy to satisfy safety with respect to any specified budget. 
    \item We propose an algorithm to compute a control policy that satisfies our derived conditions. Our proposed algorithm converges to a feasible solution, given its existence, within finite number of iterations.
    \item  We validate our proposed approach using a case study on lateral control of a Boeing 747. We show that our proposed approach guarantees the safety of Boeing 747 with respect to the given budget constraint.
\end{itemize}

The reminder of the paper is organized as follows. Section \ref{sec:related} presents the related work. Section \ref{sec:system model} introduces the CPS model and presents the problem statement. Section \ref{sec:framework} gives our proposed framework. Section \ref{sec:analysis} presents our proposed solution approach. Section \ref{sec:simulation} contains a case study on a Boeing 747. Section \ref{sec:conclusion} concludes the paper.
\section{Related Work}\label{sec:related}
Safety verification \cite{prajna2007framework,pajic2014safety} and safety controller synthesis \cite{ames2019control,cohen2020approximate,qin2021learning,herbert2017fastrack} for CPS operated in benign environment have been extensively studied. 


Fault tolerant controllers \cite{zhang2008bibliographical,sharifi2010fault,xu2020distributed} and architectures \cite{sha2001using,bak2009system,mohan2013s3a} have been widely adopted for CPS that may incur faults. One of the well-known fault tolerant designs is the simplex architecture. This architecture consists of a main controller which is vulnerable to random failures and a safety controller which is verifiable and fault-free \cite{sha2001using}. Under certain conditions, e.g., the main controller experiences a fault, a decision module instantaneously switches to the safety controller. The decision module switches back to the main controller after the main controller recovers from fault. These fault tolerant approaches assume that there is no common failure for all components which may not hold for malicious attack.

There exist two main trends of approaches to address CPS under malicious cyber attacks. The first category aims at protecting the system from malicious attacks using control- and game-theoretic approaches \cite{pajic2014robustness,fawzi2014secure,cardenas2008research}. These approaches detect the attack and then filter its impact. The second body of literature focuses on designing attack tolerant systems. The CRAs \cite{mertoguno2019physics,arroyo2019yolo,arroyo2017fired,niu2022verifying} belong to this category. 




BFT++ and other variants \cite{mertoguno2019physics} are applied when CPS have redundant controllers. One of the redundant controllers is used as backup and equipped with a buffer storing the time-delayed inputs. The non-backup controllers are deliberately engineered to crash following a malicious exploit, e.g., by implementing software diversity \cite{larsen2014sok} or memory/instruction randomization \cite{kc2003countering}. Sensing the crash, BFT++ recovers the controllers quickly from the backup whose integrity is ensured by flushing the buffer.


The CRA proposed in \cite{arroyo2019yolo} and also the restart-based mechanisms \cite{abdi2017application,abdi2018guaranteed,romagnoli2020software,arauz2021linear,niu2022verifying} are applicable to CPS that do not have redundancy. These approaches reset the cyber subsystem to a `clean' state via restart to recover from attack. The authors of \cite{arroyo2019yolo} tunes controller availability for the safety of physical subsystem, whereas the restart-based mechanisms use reachability analysis \cite{abdi2018guaranteed,romagnoli2020software,arauz2021linear} for safety guarantee.
\section{System Model and Problem Formulation} \label{sec:system model}
We first give some notations before presenting the system model. Then we state the problem investigated in this paper. 

A continuous function $\alpha:[-b,a)\rightarrow(-\infty,\infty)$ belongs to extended class $\mathcal{K}$ if it is strictly increasing and $\alpha(0)=0$ for some $a,b>0$. Throughout this paper, we use $\mathbb{R}$, $\mathbb{R}_{\geq 0}$, $\mathbb{R}_{>0}$, and $\mathbb{Z}_{\geq 0}$ to denote the set of real numbers, non-negative real numbers, positive real numbers, and non-negative integers, respectively. Given a vector $x\in\mathbb{R}^n$, we denote its $i$-th entry as $[x]_i$, where $i=1,\ldots,n$.

Consider a CPS consisting of a cyber subsystem and a physical subsystem. The physical subsystem is modeled by a plant that evolves following
\begin{equation}\label{eq:dynamics}
    \dot{x}_t = f(x_t)+g(x_t)u_t,
\end{equation}
where $x_t\in \mathcal{X} \subset \mathbb{R}^n$ is the system state and $u_t\in\mathcal{U}\subset\mathbb{R}^m$ is the control input. Functions $f:\mathbb{R}^n\rightarrow\mathbb{R}^n$ and $g:\mathbb{R}^{n}\rightarrow\mathbb{R}^{n\times m}$ are assumed to be Lipschitz continuous. We also assume that $\mathcal{U}=\prod_{i=1}^m[u_{i,min},u_{i,max}]$ with $u_{i,min}<u_{i,max}$. The physical plant is normally recommended to be operated within a certain range $\mathcal{C} = \{x \in \mathcal{X}:h(x)\geq 0\}$, where $h:\mathbb{R}^n\rightarrow\mathbb{R}$ is a continuously differentiable function. We assume that set $\mathcal{C}$ is compact. Given the system state $x$, the actuator signal $u$ is determined by a control policy $\mu:\mathcal{X}\rightarrow\mathcal{U}$.

Although the physical plant evolves in continuous time, the cyber subsystem interacts with the physical subsystem following functioning cycles. We assume that the sensors can directly measure the physical state $x$. At each functioning cycle $k\in\mathbb{Z}_{\geq 0}$, the cyber subsystem measures $x_{k\delta}$ and updates the actuator signal $u_{k\delta}=\mu(x_{k\delta})$. The actuator signal remains constant during each functioning cycle $k$. In the remainder of this paper, we refer to a functioning cycle as an epoch with length $\delta>0$.

The system is subject to a malicious attack initiated by an intelligent adversary. The adversary aims at driving the physical plant outside $\mathcal{C}$ to damage it. The adversary can exploit the vulnerabilities in the cyber subsystem and intrude into the system. Once the adversary intrudes successfully, it gains access to the software, actuators, and other peripherals. As a consequence, the actuator signal is corrupted by the adversary and deviates from the desired control policy $\mu(\cdot)$. To recover the system from attack, the CRAs and other mechanisms have been proposed, as reviewed in Section \ref{sec:intro} and \ref{sec:related}. Let $t_1\geq 0$ be a time instant when the adversary corrupts the cyber subsystem. The CRAs eliminate the adversary from the system at some time $\tilde{t}>t_1$. We denote the time instant when the adversary successfully corrupts the cyber subsystem again for the first time after $\tilde{t}$ as $t_2>\tilde{t}$. We define the interval $[t_1,t_2]$ as an \emph{attack cycle}. Note that the length $A=t_2-t_1$ of each attack cycle varies, and is dependent on the adversary. Later in Section \ref{sec:analysis}, we will compute a lower bound for $A$ to guarantee system safety. 

Due to the malicious attack, the physical plant may have to be temporarily operated outside $\mathcal{C}$. To avoid causing irreversible damage to the plant, we need to minimize the amount of time that the CPS is operated outside $\mathcal{C}$ or minimize how far the physical state $x$ deviates from $\mathcal{C}$. We capture the instantaneous damage incurred by the plant when operated outside $\mathcal{C}$ as a cost $L:\mathbb{R}\rightarrow\mathbb{R}_{\geq 0}$, defined as 
\begin{equation}\label{eq:cost}
    L(h(x)) = \begin{cases}
    L_1(-h(x)),&\mbox{ if }h(x) < 0\\
    0,&\mbox{ if }h(x)\geq 0
    \end{cases}
\end{equation}
where $L_1:\mathbb{R}_{>0}\rightarrow\mathbb{R}_{\geq 0}$ is a monotone non-decreasing function. When $L_1(-h(x))=1$, then  $\int_{t}L(x)\dd t$ is equal to the amount of time such that $x_t\notin\mathcal{C}$. When $L_1(-h(x))=-h(x)$ for all $x\notin\mathcal{C}$, Eqn. \eqref{eq:cost} models the deviation of the physical plant from the boundary of $\mathcal{C}$. We define the physical safety with respect to budget $B$ as follows.
\begin{definition}[Physical Safety with Respect to Budget $B$]\label{def:safety}
The physical plant is safe with respect to a budget $B$ if the following relation holds for any attack cycle $[t_1,t_2]$:
\begin{equation} \label{budget}
    J=\int_{t=t_1}^{t_2}L(h(x_t))\dd t\leq B.
\end{equation}
\end{definition}
Eqn. \eqref{budget} enforces an upper bound on the cost incurred by the system during any attack cycle. When $B=0$, Definition \ref{def:safety} recovers the strict safety constraint $x_t\in\mathcal{C}$ for all $t\geq 0$ as a special case. Given Definition \ref{def:safety}, the problem of synthesizing a control policy with safety guarantee is stated as follows:
\begin{problem} \label{prob1}
Synthesize a control policy  $\mu:\mathcal{X}\rightarrow\mathcal{U}$ for the CPS such that Definition \ref{def:safety} is satisfied for a given budget $B$.
\end{problem}

\section{Our Proposed Cyber Resilient Framework}\label{sec:framework}

In this section, we first detail the timing behaviors of the CRAs. Then we construct a hybrid system that models CPS adopting any of these architectures. The simplex architecture reviewed in Section \ref{sec:related} is also incorporated in our framework for completeness. We finally reformulate Problem \ref{prob1} using the constructed hybrid system.

\subsection{Timing Behaviors of the CRAs}\label{sec:architecture} 

In this subsection, we present the timing behaviors of the CRAs \cite{mertoguno2019physics,arroyo2019yolo,arroyo2017fired}, which track the status of the cyber subsystem. Note that the status of the cyber subsystem are discrete. When the adversary intrudes into the system at epoch $j$, the cyber subsystem changes from the normal to the corrupted status, indicating the adversary can arbitrarily manipulate the controllers. If the CPS have redundant controllers as discussed in \cite{mertoguno2019physics}, the non-backup controllers will crash by epoch $j+N_1$ in the worst-case, which triggers controller restoration, leading the cyber subsystem to transit from the corrupted status to the restoration status. In practical implementations of BFT++, we observe that $N_1= 2$ and the buffer length is chosen to be greater than $N_1$. Denote the worst-case number of epochs needed for controller restoration as $N_2$. Then the cyber subsystem returns to the normal status using the backup controller by epoch $j+N_1+N_2$. To ensure safety of the physical subsystem, crash delay $N_1$ needs to be small enough, and the restoration time $N_2$ needs to be tolerated by the physical subsystem's resilience $\Delta(x)$ which is determined by the physical state and system dynamics.

When CPS have no redundant controllers, cyber subsystem recovery can be triggered by either the attack or the timer \cite{arroyo2019yolo,arroyo2017fired}. If the controller crashes due to attack, which takes at most $N_3$ epochs, then the system reboots and re-initializes the controller to recover it. The worst-case time needed for such controller recovery, denoted as $N_4$, is in general larger than BFT++, i.e., $N_4 \geq N_2$. We remark that the controller restart is triggered by the timer when the attack does not crash the controller. If the recovery is triggered by the timer and there is no attack, then the system restart is executed every $N_4 + N_5$ epochs, where $N_5$ is the number of epochs elapsed in the normal status during one restart period. Safety of the physical subsystem is then ensured by tuning the time between two consecutive restarts and the controller availability.

\begin{figure}[!htp]
    \centering
    \includegraphics[scale=0.18]{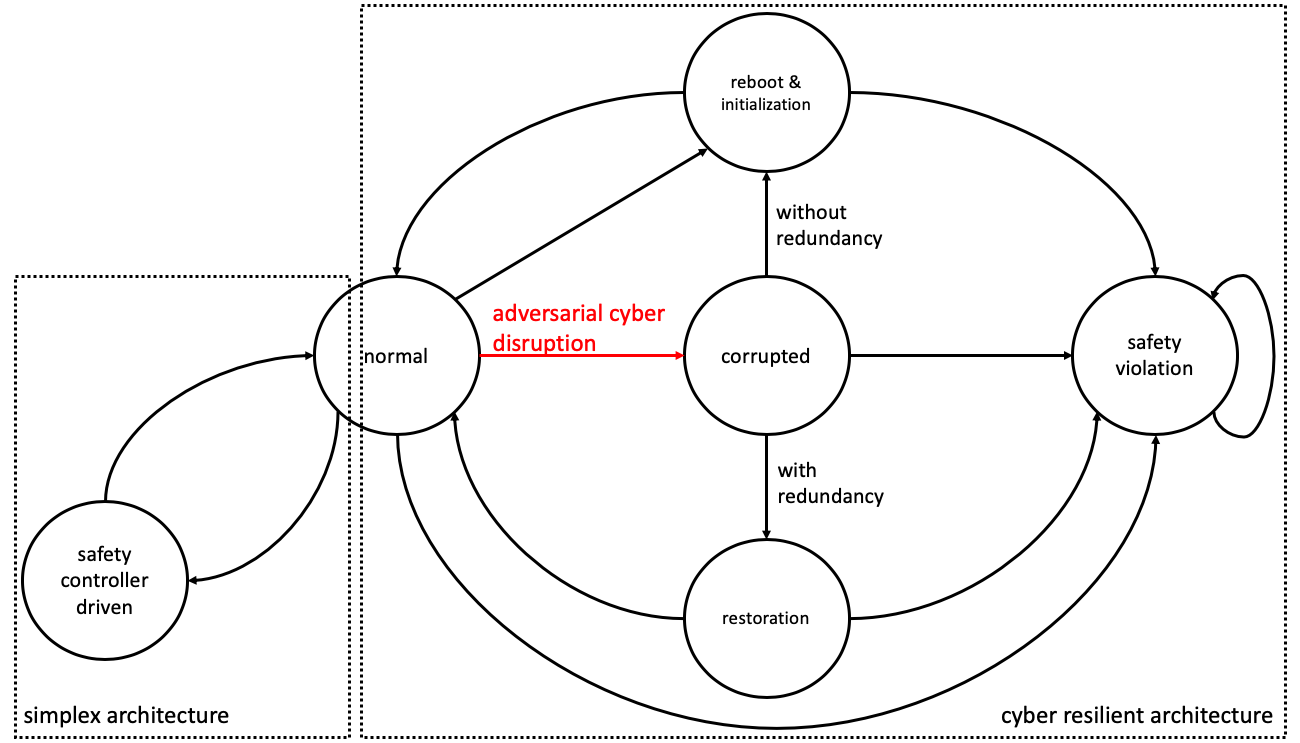}
    \caption{Hybrid system $H=(\mathcal{X},\mathcal{U},\mathcal{L},\mathcal{Y},\mathcal{Y}_0,Inv,\mathcal{F},\allowbreak \Sigma,\allowbreak\mathcal{E},\allowbreak\Phi)$ captures the simplex architecture (left part) and the CRAs (right part). Each state $y=(x,(l_1,j))\in\mathcal{Y}$ of the hybrid system captures the continuous physical state $x\in\mathcal{X}$ and the discrete location $l=(l_1,j)\in\mathcal{L}$ including the system status $l_1$ and the epoch index $j$. Each circle in the figure represents a discrete location in $\mathcal{L}$, with the epoch indices being omitted. The arrows in the figure represent the transitions $\Sigma$ of hybrid system $H$. Each transition is labeled using $e\in\mathcal{E}$, and is enabled when the corresponding clock constraint $\phi\in\Phi$ is satisfied. The detailed labels and clock constraints are given in Table \ref{table:transition}. The transition in red color is triggered by external event, i.e., the adversarial cyber disruption. }
    \label{fig:unify framework}
\end{figure}

\begin{table*}[!b] 
    \centering
    \begin{tabular}{|m{13mm}|c|c|c|m{33mm}|}
    
    \hline
    Designs &Starting state $(x,(l_1,j))$  &  Target state $(x',(l_1',j'))$ & Label $e=(\gamma,e_2)$ & Clock constraint $\phi$\\
    \hline
    \hline
   \multirow{2}{*}{\shortstack[l]{Simplex \\architecture}} &  $(x,(normal,j))$    & $(x,(SC,j'))$ & $(Safety~controller~activated,e_2)$ & $ (j+e_2=j')\land (e_2=0)$\\
     \cline{2-5}
   &  $(x,(SC,j))$    & $(x,(normal,j'))$ & $(Safety~controller ~inactivated,e_2)$ & $(j+e_2=j')\land (e_2=0)$\\
     \hline
    \multirow{7}{*}{\shortstack[l]{the CRAs}} &  $(x,(normal,j))$    & $(x',(corrupted,j'))$ & $(Adversarial~ cyber~ disruption,e_2)$ & $ (j+e_2=j')\land (e_2\geq 0)$\\
     \cline{2-5}
    &  $(x,(corrupted,j))$    & $(x',(restoration,j'))$ & $(Controller~crash,e_2)$ & $ (j+e_2=j')\land (e_2\leq N_1)$\\
    \cline{2-5}
     &  $(x,(restoration,j))$    & $(x',(normal,j'))$ & $(Controller~restored,e_2)$ & $ (j+e_2=j')\land (e_2\leq N_2)$\\
      \cline{2-5}
     &  $(x,(normal,j))$    & $(x',(R\&I,j'))$ & $(Timer,e_2)$ & $ (j+e_2=j')\land (e_2=N_5)$\\
     \cline{2-5}
     &  $(x,(corrupted,j))$    & $(x',(R\&I,j'))$ & $(Controller~crash,e_2)$ & $ (j+e_2=j')\land (e_2\leq N_3)$\\
     \cline{2-5}
     &  $(x,(corrupted,j))$    & $(x',(R\&I,j'))$ & $( Timer,e_2)$ & $ (j+e_2=j')\land (e_2\leq N_5)$\\
     \cline{2-5}
     &  $(x,(R\&I,j))$    & $(x',(normal,j'))$ & $(Initialization~done,e_2)$ & $ (j+e_2=j')\land (e_2\leq N_4)$\\
     \cline{2-5}
     &  $(x,(\cdot,j))$    & $(x',(safety~violation,j'))$ & $(Maximum ~tolerance,e_2)$ & $ (j+e_2=j')\land (e_2\geq \Delta(x))$\\
     \hline
    \end{tabular}
    \caption{This table shows the transitions with their corresponding labels and clock constraints. The second and third columns give the starting and end states of a transition, respectively. The fourth column presents the label $e=(\gamma,e_2)$ associated with the transition $\sigma=((x,(l_1,j)),(x',(l_1',j')))$. The trigger event of the transition is denoted as $\gamma$, and the time elapsed in $l_1$ is denoted using $e_2\in\mathbb{Z}_{\geq 0}$. The fifth column gives the clock constraint $\phi$ that needs to be satisfied by the transition $\sigma$ and parameter $e_2$. Parameters $N_1,N_2,N_3,N_4$, and $N_5$ are determined by the architecture design. Parameter $\Delta(x)$ captures the maximum tolerance provided by the physical subsystem.}
    \label{table:transition}
\end{table*}

\subsection{Our Proposed Framework}

In this subsection, we first construct a hybrid system to model the CPS implementing the simplex architecture or CRAs. We then restate Problem \ref{prob1} in the context of the hybrid system. We construct a hybrid system $H=(\mathcal{X},\mathcal{U},\mathcal{L},\mathcal{Y},\mathcal{Y}_0,Inv,\mathcal{F},\allowbreak \Sigma,\mathcal{E},\Phi)$, as shown in Fig. \ref{fig:unify framework}, where
\begin{itemize}
    \item $\mathcal{X}\subseteq\mathbb{R}^{n}$ is the continuous state space modeling the states of the physical subsystem. $\mathcal{U}\subseteq\mathbb{R}^m$ is the set of admissible control inputs of the physical subsystem.
    \item $\mathcal{L}=\{normal, R\& I, SC, restoration,\allowbreak corrupted, \allowbreak safety~violation\}\times\mathbb{Z}_{\geq 0}$ is a set of discrete locations\footnote{Throughout this paper, we denote $l=reboot\&initialization$ as $l=R\&I$ and denote $l=safety~controller~driven$ as $l=SC$ for simplicity.}, with each location $l\in\mathcal{L}$ modeling the status of the system at each epoch index. 
    \item $\mathcal{Y}=\mathcal{X}\times\mathcal{L}$ is the state space of hybrid system $H$, and $\mathcal{Y}_0\subseteq \mathcal{Y}$ is the set of initial states.
    \item $Inv:\mathcal{L}\rightarrow 2^\mathcal{X}$ is the invariant that maps from the set of locations to the power set of $\mathcal{X}$. That is, $Inv(l)\subseteq\mathcal{X}$ specifies the set of possible continuous states when the system is at location $l$.
    \item $\mathcal{F}$ is the set of vector fields. For each $F\in\mathcal{F}$, the continuous system state evolves as $\dot{x}=F(x,u,l)$, where $F$ is jointly determined by the system dynamics and the status of the cyber subsystem, and $\dot{x}$ is the time derivative of continuous state $x$.
    \item $\Sigma\subseteq \mathcal{Y}\times \mathcal{Y}$ is the set of transitions between the states of the hybrid system. A transition $\sigma=((x,l),(x',l'))$ models the state transition from $(x,l)$ to $(x',l')$.
    \item $\mathcal{E}=\Gamma\cup\mathbb{Z}_{\geq 0}$ is a set of labels, where $\Gamma$ is the finite alphabet set. Each $\gamma\in\Gamma$ is labeled on some transition $\sigma\in \Sigma$ indicating the events that triggers the transition.
    \item $\Phi$ is a set of clock constraints, with each $\phi\in\Phi$ being defined as $\phi: \Sigma\times \mathbb{Z}_{\geq 0}\rightarrow\{0,1\}$. Function $\phi$ maps the time elapsed in each discrete location labeled on each transition to the binary set $\{0,1\}$, indicating if the transition is enabled or not.
\end{itemize}
In Fig. \ref{fig:unify framework}, we only label the first element of each location $l=(l_1,j)$, where $l_1\in\{normal, R\& I, SC, restoration,\allowbreak corrupted, \allowbreak safety~violation\}$, and the epoch index $j\in\mathbb{Z}_{\geq 0}$ is omitted. Variable $l_1$ represents the status of the system, as explained in Section \ref{sec:architecture}. Particularly, we use $R\&I$ to represent controller reboot and initialization for CPS without redundancy, and use $SC$ to represent the status where the system is driven by the safety controller as suggested in the simplex architecture. In the remainder of this paper, we refer to $l_1$ as the location of $H$ omitting the epoch index when the context is clear. The set of vector fields $\mathcal{F}$ captures the dynamics at each discrete location. For instance, when $l=(R\&I,j)$, we have that $\dot{x}=F(x,u,l)=f(x)$ for all $j\in\mathbb{Z}_{\geq 0}$ since $u_t=0$.

We label each transition $\sigma=((x,l),(x',l'))\in\Sigma$ using $e=(\gamma,e_2)\in\mathcal{E}$. The detailed label associated with each transition can be found in Table \ref{table:transition}. Here we use $\gamma$ to represent the event that triggers the transition. For instance, the transition from $normal$ to $corrupted$ is triggered by the adversarial cyber disruption, whereas the transition from $corrupted$ to $restoration$ is triggered by controller crash. The element $e_2\in\mathbb{Z}_{\geq 0}$ denotes the number of epochs elapsed in status $l_1$ before the occurrence of transition $\sigma=((x,(l_1,j)),(x',(l_1',j')))$. For example, at most $N_1$ epochs elapse at location $corrupted$ before the transition from $(x,(corrupted,j))$ to $(x',(restoration,j'))$ occurs. When $e_2=0$, it indicates that the transition occurs instantaneously. 

A transition in hybrid system $H$ is enabled if and only if a clock constraint $\phi$ associated with the transition is satisfied. Consider a transition $\sigma=((x,(l_1,j)),(x',(l_1',j')))$ labeled with $e=(\gamma,e_2)$. A clock constraint $\phi$ verifies if $j$, $j'$, and $e_2$ satisfies $j+e_2=j'$ with $e_2$ being determined by the architecture. The clock constraint enables hybrid system $H$ to always track the correct epoch index.


We remark that for each status $l_1$, we do not depict the transitions $(x,(l_1,j))$ to $(x,(l_1,j+1))$ for compactness of the figure. These transitions do not cause any system status jump, and only track the evolution of epoch indices. We also define a guard set $\mathcal{G}(l,l')$ as $\mathcal{G}(l,l')=\{x \in \mathcal{X}: ((x,l),(x,l')) \in \Sigma \}$ which represents the set of physical states starting from which the system can transit from location $l$ to $l'$ on the hybrid system $H$.


The transitions that end at $l_1=safety~violation$ are triggered by event $\gamma=Maximum~tolerance$, indicating that the physical subsystem has not received the correct input in time and has utilized all resilience against the disruption. In this case, safety violation $J>B$ becomes inevitable, and the system cannot recover after safety violation (captured via the self-loop in Fig. \ref{fig:unify framework}). We capture the maximum tolerance provided by the physical subsystem using $e_2$ labeled on the transitions. Note that the tolerance depends on the physical system state $x$ and is denoted as $\Delta(x)$. 



We are now ready to translate problem \ref{prob1} using the context of hybrid system $H$ as follows:\\
\textbf{Restatement of Problem \ref{prob1}.} Given hybrid system $H$, synthesize a control policy $\mu:\mathcal{X}\rightarrow\mathcal{U}$ such that hybrid system $H$ never reaches status $l_1=safety~violation$.

\section{Analysis of the Proposed Framework} \label{sec:analysis}

This section presents the proposed solution approach to Problem \ref{prob1}. We first develop sufficient conditions for the control policy that guarantees safety of the physical subsystem under a cyber attack. Then we formulate the derived conditions as sum-of-squares (SOS) constraints and propose an algorithm to compute a control policy and the corresponding parameters. We finally give the convergence and complexity of our algorithm.

In the following, we derive the sufficient conditions under which a control policy $\mu:\mathcal{X}\rightarrow\mathcal{U}$ ensures the system to satisfy Definition \ref{def:safety} with respect to a given budget $B$. The idea is that if control policy $\mu$ drives the physical subsystem to $\mathcal{C}_1=\{x:h(x)\geq c\}\subseteq\mathcal{C}$ when $l_1=normal$ and we can constrain the system trajectory to remain in a set $\mathcal{D}=\{x:h(x)\geq -d\}\supseteq\mathcal{C}$ for any $l_1\in\{corrupted, R\&I,restoration\}$, then we can limit the worst-case cost incurred during one attack cycle to be bounded by $B$ by tuning choices $c,d\geq0$. We denote the worst-case number of epochs when the system is at some status $l_1\in\{corrupted, R\&I,restoration\}$ as $N$. We then have the following conditions:

\begin{theorem}\label{thm:sufficient condition}
Consider hybrid system $H$ and let set $\mathcal{C}$ be defined as in Section \ref{sec:system model}. Let $h_c(x) = h(x)-c$ and $h_d(x) = h(x) + d$. We define $\mathcal{C}_1=\{x:h_c(x)\geq 0\}$ and $\mathcal{D}=\{x:h_d(x)\geq 0\}$. Consider an arbitrary attack cycle denoted as $[t_1,t_2]$ and suppose $x_{t_1}\in\mathcal{C}_1$. If there exist constants $c,d\geq0$, $\tau>0$, and a control policy $\mu:\mathcal{X}\rightarrow\mathcal{U}$ such that
\begin{subequations}\label{eq:time derivative BC}
\begin{align}
    &\frac{\partial h_d}{\partial x}(x)(f(x)+g(x)u)\geq -\frac{c+d}{N\delta},~\forall (x,u)\in\mathcal{D}\times\mathcal{U}\label{eq:time derivative BC 1}\\
    &\frac{\partial h_c}{\partial x}(x)(f(x)+g(x)\mu(x))\geq \frac{c+d}{\tau},~\forall x\in\mathcal{D}\setminus \mathcal{C}_1  \label{eq:time derivative BC 3}\\ 
    &\frac{\partial h_c}{\partial x}(x)(f(x)+g(x)\mu(x))\geq -\alpha(h_c(x)),~\forall x\in\mathcal{C}_1  \label{eq:time derivative BC 4}\\
    & \frac{N\delta}{c+d}\int_{s=0}^{d}L_1(s)\dd s + \frac{\tau}{c+d} \int_{s=0}^{d} L_1(s)\dd s \leq B  \label{eq:time derivative BC 2}
\end{align}
\end{subequations}
then system \eqref{eq:dynamics} is safe with respect to budget $B$ by taking policy $\mu$ at $l_1=normal$, provided that $A=t_2-t_1 \geq \tau+N\delta$. Furthermore, $x_t \in \mathcal{D}$  for $t \in [t_1,t_2]$ and $x_t \in \mathcal{C}_1$  for $t \in [t_1+\tau+N\delta,t_2]$.
\end{theorem}
\begin{proof}
The proof consists of two steps. We first find the guard set for each state of hybrid system $H$ when the conditions in Eqn. \eqref{eq:time derivative BC} hold. We then prove that the system is safe with respect to budget $B$ according to Definition \ref{def:safety}.

In the first step, we show that if $t_2\geq t_1+\tau+N\delta$, $x_{t_1}\in\mathcal{C}_1$ and the conditions in (\ref{eq:time derivative BC}) hold, then $x_t \in \mathcal{D}$ for $t \in [t_1,t_2]$. Additionally we show that $x_t \in \mathcal{C}_1$  for $t \in [t_1+\tau+N\delta,t_2]$ using control policy $\mu$ at $l_1=normal$. As hybrid system $H$ takes $N$ epochs to transition to $l_1=normal$ after being corrupted, therefore $t_1+N\delta$ is the time instant when transition $\sigma=(\cdot,normal)$ takes place after $t_1$. Then for any $t' \in [t_1,t_1+N\delta]$ and for any $u \in \mathcal{U}$, we have that
\begin{align}
h(x_{t'})= h(x_{t_1})&+\int_{t=t_1}^{t'}\dot{h}\dd t \geq  c - \frac{c+d}{N\delta} (t'-t_1) \geq -d, \label{label_set_corrupt}
\end{align}
where the inequality holds by Eqn. \eqref{eq:time derivative BC 1} and the assumptions that $x_{t_1}\in\mathcal{C}_1$ and $t' \in [t_1,t_1+N\delta]$. Therefore, we have that $x_{t'} \in \mathcal{D}$ for all $t' \in [t_1,t_1+N\delta]$, indicating that $\mathcal{G}(R\&I,normal), \mathcal{G}(restoration,normal)\subseteq\mathcal{D}$.

Now consider that control policy $\mu$ is applied when the system is at $l_1=normal$. Let $\hat{t}$ be any time when the system is at $l_1=normal$ for which the trajectory remains in $\mathcal{D}\setminus \mathcal{C}_1$ (i.e. $-d \leq h(x_{\hat{t}}) <-c)$. Then we can write
\begin{multline} \label{label_set_normal}
    h(x_{\hat{t}})=h(x_{t_1+N\delta})+\int_{t=t_1+N\delta}^{\hat{t}}\dot{h}\dd t\\
    \geq -d+\frac{c+d}{\tau}(\hat{t} -t_1-N\delta),
\end{multline} 
where the inequality holds by Eqn. \eqref{eq:time derivative BC 3} and the fact that $\mathcal{G}(R\&I,normal), \mathcal{G}(restoration,normal)\subseteq\mathcal{D}$. If $\hat{t} \geq t_1+N\delta+\tau$, then Eqn. \eqref{label_set_normal} renders $h(x_{\hat{t}})\geq c$ and thus $x_{\hat{t}}\in\mathcal{C}_1$. Further note that by \cite[Thm. 2]{ames2019control}, $\mathcal{C}_1$ is forward invariant using control policy $\mu$. Thus, $x_t \in \mathcal{C}_1,~\forall t \in [t_1+N\delta+\tau,t_2]$. Since $\mathcal{C}_1\subseteq\mathcal{D}$, we have $x_t \in \mathcal{D}$  for all $t \in [t_1,t_2]$.

In the second step, we quantify the worst-case cost incurred by the system. We first compute the cost incurred during time $[t_1,t_1+N\delta]$. Suppose the system reaches the boundary of $\mathcal{C}$ at time instants $z_1\leq z_2\leq \ldots\leq z_K\leq t_1+N\delta$, where $z_1\geq t_1$ and $z_K\leq t_1+N\delta$. We have that
\begin{multline*}
    J_1\doteq\int_{t=t_1}^{t_1+N\delta}L(h(x_t))\dd t
    =\int_{t=t_1}^{z_1}L(h(x_t))\dd t\\
    +\sum_{j=1}^{K-1}\int_{t=z_{j}}^{z_{j+1}}L(h(x_t))\dd t
    +\int_{t=z_K}^{t_1+N\delta}L(h(x_t))\dd t
\end{multline*}
Since $h(x_t)\in\mathcal{C}$ for $t\in[t_1,z_1]$, therefore $L(h(x_t))=0$. Additionally, since $L_1$ is monotone non-decreasing, by Eqn. \eqref{label_set_corrupt} we have that $L(h(x_t))=L_1(-h(x_t))\leq L_1(-c+\frac{c+d}{N\delta}(t-t_1))$ for any $t\in[z_j,z_{j+1}]$ if $h(x_t)\notin\mathcal{C}$. Using these arguments, we have that 
\begin{align*}
    J_1&\leq \sum_{j=1}^{K-1}\int_{t=z_j}^{z_{j+1}}L_1(-c+\frac{c+d}{N\delta }t)\dd t\\
    &\quad\quad\quad\quad\quad\quad+\int_{t=z_K}^{t_1+N\delta}L_1(-c+\frac{c+d}{N\delta }t)\dd t\\
    &=\int_{t=z_1}^{t_1+N\delta}L_1(-c+\frac{c+d}{N\delta }t)\dd t.
\end{align*}
Using Eqn. \eqref{label_set_corrupt}, it follows that $z_1\geq t_1+c/(\frac{c+d}{N\delta})=t_1+\frac{cN\delta}{c+d}$. Therefore, we have that 
\begin{align*}
    J_1 &\leq \int_{t=t_1+\frac{cN\delta}{c+d}}^{t_1+N\delta}L_1(-c+\frac{c+d}{N\delta }(t-t_1))\dd t \\
    &= \int_{t=0}^{\frac{d N\delta}{c+d}}L_1(\frac{c+d}{N\delta }t)\dd t= \frac{N\delta}{c+d}\int_{s=0}^{d}L_1(s)\dd s,
\end{align*}
where the above holds by variable substitution and the fact that $L_1$ is non-negative.

We now quantify the worst-case cost incurred during time $[t_1+N\delta,t_2]$. By Eqn. \eqref{label_set_normal}, $h(x_t)\geq -d+\frac{c+d}{\tau}(t-t_1-N\delta)$ for $x_t \in \mathcal{D}\setminus \mathcal{C}_1$. Note that $h(x_t) \geq 0$ for all $t\in[t_1+N\delta+\frac{d\tau}{c+d},t_2]$ using control policy $\mu(x)$. Therefore we have that
\begin{align*}
    J_2&\doteq\int_{t=t_1+N\delta}^{t_2}L(h(x_t))\dd t\\
    &=\int_{t=t_1+N\delta}^{t_1+N\delta+\frac{d\tau}{c+d}}L(h(x_t))\dd t+\int_{t=t_1+N\delta+\frac{d\tau}{c+d}}^{t_2}L(h(x_t))\dd t\\
    & \leq \int_{t=N\delta}^{N\delta+\frac{d\tau}{c+d}} L_1(d-\frac{c+d}{\tau}(t-N\delta))\dd t \\
    & = \int_{t=0}^{\frac{d\tau}{c+d}} L_1(d-\frac{c+d}{\tau}t)\dd t =\frac{\tau}{c+d} \int_{s=0}^{d} L_1(s)\dd s
\end{align*}
where the above holds by (\ref{label_set_normal}), $L_1$ is monotone non-decreasing and non-negative, and $L(h(x_t))=0$ for $x_t \in \mathcal{C}$. 
Therefore we have that the total cost is upper bounded by $J_1+J_2$, which yields condition \eqref{eq:time derivative BC 2}. 
\end{proof}

The above result also encompasses the case when the physical subsystem is subject to a strict safety constraint, i.e., $x_t\in\mathcal{C}$ for all $t\geq 0$. This case can be captured by letting $B=0$ in Definition \ref{def:safety}. The sufficient conditions for a control policy with strict safety guarantee are given as follows.
\begin{corollary}
Consider hybrid system $H$ and a safety set $\mathcal{C}$. Let $h_c(x) = h(x)-c$ and $\mathcal{C}_1=\{x:h_c(x)\geq 0\}$. Consider an arbitrary attack cycle denoted as $[t_1,t_2]$ and suppose $x_{0}\in\mathcal{C}_1$. If there exist constants $c\geq0$, $\tau>0$, and a control policy $\mu:\mathcal{X}\rightarrow\mathcal{U}$ such that
\begin{subequations}
\begin{align}
    &\frac{\partial h}{\partial x}(x)(f(x)+g(x)u)\geq -\frac{c}{N\delta},~\forall (x,u)\in\mathcal{C}\times\mathcal{U}\label{eq:strict safety 1}\\
    &\frac{\partial h_c}{\partial x}(x)(f(x)+g(x)\mu(x))\geq \frac{c}{\tau},~\forall x\in\mathcal{C}\setminus \mathcal{C}_1 \label{eq:strict safety 2}\\ 
    &\frac{\partial h_c}{\partial x}(x)(f(x)+g(x)\mu(x))\geq -\alpha(h_c(x)),~\forall x\in\mathcal{C}_1 \label{eq:strict safety 3}
\end{align}
\end{subequations}
$x_t \in \mathcal{C}$ for all $t \geq0$ provided that $A=t_2-t_1 \geq \tau+N\delta$.
\end{corollary}
\begin{proof}
The corollary can be proved as a special case of Theorem \ref{thm:sufficient condition} with $d=B=0$, which yields that $\mathcal{D}=\mathcal{C}$ and thereby $x_t\in\mathcal{C},~\forall t\in[t_1,t_2]$. Note that $[t_1,t_2]$ is an arbitrary attack cycle and $x_0\in\mathcal{C}_1$, rendering $x_t\in\mathcal{C},~\forall t\geq 0$.  
\end{proof}

The above analysis can also be used for the safety controller design of the simplex architecture. The safety controller, which is invoked when the system approaches the boundary of $\mathcal{C}$, can be obtained using Theorem \ref{thm:sufficient condition} by letting $c=d=B=0$ to guarantee strict safety with respect to $\mathcal{C}$. Control policy $\mu$ only needs to satisfy Eqn. \eqref{eq:time derivative BC 4} in this case.

Now we focus on the computation of control policy $\mu$ as well as parameters $c,d\geq 0$ and $\tau>0$ so that safety is satisfied according to Definition \ref{def:safety}. Our idea is to translate the conditions in Theorem \ref{thm:sufficient condition} to a set of sum-of-squares (SOS) constraints. We first make the following assumption.
\begin{assumption}\label{assump:semi-algebraic}
We assume that functions $f(x)$, $g(x)$, and $h(x)$ are polynomial in $x$. Additionally, we assume that function $L_1$ is polynomial in $-h(x)$. 
\end{assumption}
When Assumption \ref{assump:semi-algebraic} holds, $L_1(-h(x))$ can be written as $L_1(-h(x))=\sum_{i=0}^k (-h(x))^i a_i$, where $a_i$ is the coefficient of $(-h(x))^i$ for each $i=0,\ldots,k$. Next we formulate Eqn. \eqref{eq:time derivative BC} as a set of SOS constraints. 
\begin{proposition}\label{prop:sos}
Suppose there exist parameters $c,d\geq 0$ and $\theta>0$ such that the following expressions are SOS:
\begin{subequations}\label{eq:sos}
\begin{align}
    &\frac{\partial h_d}{\partial x}(x)[f(x)+g(x)u] +\frac{c+d}{N\delta} - q(x,u)h_d(x)\label{eq:sos 1}\\
    &-\sum_{i=1}^m(w_i(x,u)([u]_i-[u]_{i,min})+v_i(x,u)([u]_{i,max}-[u]_i)),\nonumber \\
    &\frac{\partial h_c}{\partial x}(x)[f(x)+g(x)\lambda(x)]-(c+d)\theta\nonumber\\
    &\quad\quad\quad\quad\quad\quad\quad\quad- l(x)h_d(x) + p(x)h_c(x),\label{eq:sos 2}\\
    &\frac{\partial h_c}{\partial x}(x)[f(x)+g(x)\lambda(x)]+\alpha(h_c(x)) -r(x) h_c(x),\label{eq:sos 3} \\
    &\lambda_i(x)-[u]_{i,min},\quad [u]_{i,max}-\lambda_i(x),~\forall i=1,\ldots,m, \label{eq:sos 4}
\end{align}
\end{subequations}
and the following inequality holds:
\begin{equation}
B(c+d) - (N\delta+\frac{1}{\theta})\sum_{i=1}^k\frac{a_id^{i+1}}{i+1} \geq 0\label{eq_sos_pos}    
\end{equation}
where $l(x),p(x),q(x,u),r(x)$ are SOS, $\lambda_i(x)$ is a polynomial in $x$ for each $i=1,\ldots,m$, and $w_i(x,u)$ and $v_i(x,u)$ are SOS for each $i=1,\ldots,m$. Then $\mu(x)=\lambda(x)=[\lambda_1(x),\ldots,\lambda_m(x)]^\top, c,d$, and $\tau=\frac{1}{\theta}$ satisfy the conditions in Eqn. \eqref{eq:time derivative BC}.
\end{proposition}
\begin{proof}
Consider $x\in\mathcal{D}$ and $[u]_{i,min}\leq [u]_i\leq [u]_{i,max}$ for all $i=1,\ldots,m$. Then we have that $h_d(x)\geq 0$, $[u]_i-[u]_{i,min}\geq 0$, and $[u]_{i,max}-[u]_i\geq 0$. Since expression \eqref{eq:sos 1}, $q(x,u)$, $w_i(x,u)$, and $v_i(x,u)$ are SOS for all $i=1,\ldots, m$, therefore for all $(x,u)\in\mathcal{D}\times\mathcal{U}$ we can write 
\begin{multline*}
    \frac{\partial h_d}{\partial x}(x)[f(x)+g(x)u] +\frac{c+d}{N\delta} \geq  q(x,u)h_d(x)\nonumber\\
    +\sum_{i=1}^m(w_i(x,u)([u]_i-[u]_{i,min})+v_i(x,u)([u]_{i,max}-[u]_i))\geq 0. 
\end{multline*}
Thus condition \eqref{eq:time derivative BC 1} holds.

Expressions \eqref{eq:sos 2} to \eqref{eq:sos 4} can be proved similarly. Eqn. \eqref{eq_sos_pos} follows by computing the integrals in Eqn. \eqref{eq:time derivative BC 2}. Details are omitted due to space constraint.
\end{proof}
  \begin{center}
  	\begin{algorithm}[!htp]
  		\caption{Heuristic algorithm for computing $c$, $d$, $\tau$ and control policy $\mu(x)$}
  		\label{algo:max range}
  		\begin{algorithmic}[1]
  			\State \textbf{Input}: $f(x)$, $g(x)$, $B$, $\tau_{max}$, $c_{max}$, $\epsilon_1>0$, $\epsilon_2>0$ 
  			\State \textbf{Output:} $c$, $d$, $\tau$, $\lambda(x)$
  			\State \textbf{Initialization:} $c=0$.
  		    \While{$c\leq c_{max}$}
  		    \State $d=0$
            \While{$d \leq d_{max}$}
            \State Maximize ${\theta}$ subject to \eqref{eq:sos} with $c$ and $d$ fixed.
            \If{Eqn. \eqref{eq:sos} is feasible, \eqref{eq_sos_pos} is satisfied and $\frac{1}{\theta} < \tau_{max}$}
            \State \textbf{return} $d$, $c$, $\tau=\frac{1}{\theta}$, and $\lambda(x)$
            \Else
            \State $d=d+\epsilon_1$
            \EndIf
            \EndWhile
            \State $c=c+\epsilon_2$
            \EndWhile
  		\end{algorithmic}
  	\end{algorithm}
  \end{center}

Simultaneously searching for $\lambda(x), c, d$ and $\theta$ that satisfy Proposition \ref{prop:sos} leads to bilinearity in \eqref{eq:sos}. To this end, we propose an algorithm to compute $\lambda(x), c, d$ and $\theta$ that satisfy Proposition \ref{prop:sos}, as shown in Algorithm \ref{algo:max range}. Algorithm \ref{algo:max range} first initializes parameters $c=d=0$. At each iteration, the algorithm maximizes $\theta$ using the given $c$ and $d$. If some $\theta^*$ can be found in line 7 which satisfies the conditions in line 8 ($\tau_{max} = \infty$ if not specified), then the algorithm returns $c$, $d$, and set $\tau=\frac{1}{\theta^*}$ and $\mu(x)=\lambda(x)$. Otherwise, the algorithm increases the values of parameters $c$ and/or $d$ and repeat the search process for parameter $\theta$. Algorithm \ref{algo:max range} terminates at $c=c_{max}$ and $d=d_{max}$ if no feasible solution to \eqref{eq:sos} and \eqref{eq_sos_pos} is found, where $c_{max}=\sup_{x\in\mathcal{C}}h(x)$ and $d_{max}$ is the maximum value of $d$ that satisfies $\sum_{i=0}^kN\delta\frac{a_id^{i+1}}{i+1} \leq (c+d)B$.

Now we briefly characterize the convergence and complexity of Algorithm \ref{algo:max range}. Our intuition is that if there exists a feasible solution satisfying Eqn. \eqref{eq:sos} and \eqref{eq_sos_pos} strictly, then this solution must lie within the interior of the feasible solution set. Thus by choosing $\epsilon_1$ and $\epsilon_2$ appropriately small, the convergence of Algorithm \ref{algo:max range} can be guaranteed. We formalize this convergence result in the following proposition.

\begin{proposition} \label{convergence}

Suppose the set $\mathcal{C}$ is compact and $L_1$ is polynomial in $-h(x)$ with non-zero degree. Further assume that the feasible solution $(c,d)$ for \eqref{eq:sos} and \eqref{eq_sos_pos} satisfy inequality constraints in  \eqref{eq:sos} and \eqref{eq_sos_pos} strictly with $0 \leq c \leq c_{max}$, $0 \leq d \leq d_{max}$ and $0 < \tau \leq \tau_{max}$. Then Algorithm \ref{algo:max range} finds a feasible solution with $0 < \tau \leq \tau_{max}$ in finite number of iterations if $\epsilon_1$ and $\epsilon_2$ are chosen appropriately small.

\end{proposition}

\begin{proof}

Since $\mathcal{C}$ is compact, we have $c_{max}=\sup_{x\in\mathcal{C}}h(x)<\infty$. Also, since the order of $d$ in $\sum_{i=0}^kN\delta\frac{a_id^{i+1}}{i+1}$ is greater than that of $(c+d)B$ and $L_1$ is non-negative, therefore $d_{max}<\infty$. Using $c_{max},d_{max}<\infty$, we have that Algorithm \ref{algo:max range} terminates in finite number of iterations. 

Let $(c,d)$ satisfy \eqref{eq:sos} and \eqref{eq_sos_pos} strictly with $0 \leq c \leq c_{max}$, $0 \leq d \leq d_{max}$, and $0 < \tau \leq \tau_{max}$. Therefore there exists an interval $\mathcal{I}\subseteq\mathbb{R}^2$ for $(c,d)$ with non-zero measure for which \eqref{eq:sos} and \eqref{eq_sos_pos} are feasible and $ c \in[0, c_{max}]$, $d \in[0, d_{max}]$, $ \tau \in(0,\tau_{max}]$. Denote the length of $\mathcal{I}$ in $c$ and $d$ as $\tilde{c}>0$ and $\tilde{d}>0$, respectively. Let $\epsilon_1 \in(0, \tilde{c})$ and $\epsilon_2 \in(0, \tilde{d})$. Then Algorithm \ref{algo:max range} will always terminate with a feasible solution. Otherwise interval $\mathcal{I}$ contains some infeasible solutions to Eqn. \eqref{eq:sos} and \eqref{eq_sos_pos}, which contradicts its definition.
\end{proof}

By the proof of Proposition \ref{convergence}, the computational complexity of Algorithm \ref{algo:max range} is $\lfloor \frac{c_{max}} {\epsilon_1} \rfloor \lfloor \frac{d_{max}} {\epsilon_2} \rfloor M$, where $M$ is the computational complexity of line 7 in Algorithm \ref{algo:max range}. 

\section{Case Studies}\label{sec:simulation}

This section presents a case study on lateral control of a Boeing 747. The lateral dynamics for a Boeing 747 (at Mach 0.8 and 40000ft) \cite{franklin2002feedback} are given as $\dot{x}=f(x)+g(x)u$, where
\begin{equation*}
    f(x)=\begin{bmatrix}
    -0.0558 & -0.9968 & 0.0802 & 0.0415\\
    0.598 & -0.115 & -0.0318 & 0\\
    -3.05 & 0.388 & -0.465 & 0\\
    0 & 0.0805 & 1 & 0
    \end{bmatrix},
\end{equation*}
$g(x)=[0.00729,-0.475,0.153,0]^\top$, $x\in\mathbb{R}^4$ with $x_1,x_2,x_3$, and $x_4$ representing the side-slip angle, yaw rate, roll rate, and roll angle, respectively. The cyber subsystem updates the control signal with frequency $20Hz$. The aircraft aims at maintaining the yaw rate $x_2$ within $\mathcal{C}=\{x:h(x)\geq 0\}$ for passengers' comfort and minimizing potential damage to baggage, where $h(x)=0.025^2-x_2^2$. Set $\mathcal{C}$ is represented as the white region in Fig. \ref{fig:yaw rate}. We set $x_0=[0.01,0.025,0,0]^\top$, $B=0.02$, and study the following three scenarios.

\emph{\underline{Scenario I: The aircraft has redundant controllers.}} We choose parameter $N=N_1+N_2$ epochs during which the aircraft is either corrupted or under restoration, where $N_1=N_2=2$ \cite{mertoguno2019physics}. We consider the system is equipped with a buffer of length $3$. Using Algorithm \ref{algo:max range}, we obtain that $c=0$, $d=0.4$, and $\tau=0.16s$, indicating that we need $4$ epochs for the system to be at location $l_1=normal$ after restoration. We let the length of each attack cycle be $8$ epochs. The control policy given by Algorithm \ref{algo:max range} is $u=Kx$ with $K=[- 9.231\times 10^{-3}, 0.503, - 1.805\times 10^{-3}, 2.373\times 10^{-5}]$. We depict the trajectory of the yaw rate over $150$ epochs using the black solid line in Fig. \ref{fig:yaw rate}. The non-smoothness in the trajectory is due to switching between location $normal$ (the controller is available) and locations $corrupted$ and $restoration$ (the controller is unavailable). The adversary enforces the yaw rate to exceed $0.025$ from the second to fifth epoch with cost $0.0038 < B$.

\begin{figure}[!htp]
    \centering
    \includegraphics[scale = 0.45]{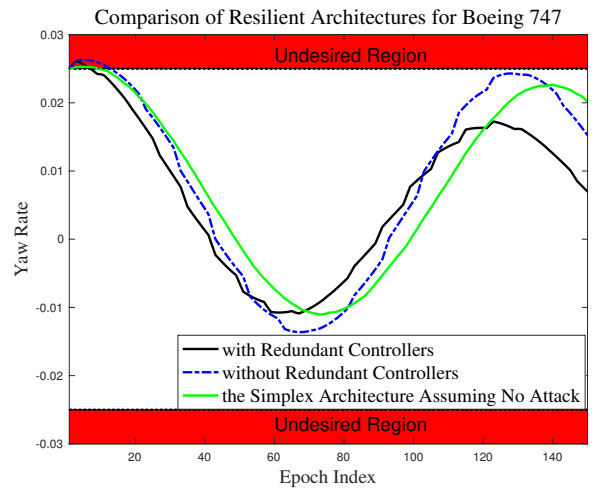}
    \caption{The yaw rate of a Boeing 747 adopting different architectures over $150$ epochs. Set $\mathcal{C}=\{x_2:0.025^2-x^2\geq 0\}$ represented by the white region. The black solid line depicts the yaw rate of an aircraft with redundancy. The blue dash-dotted line is the yaw rate of an aircraft without redundancy. The green solid line describes the yaw rate of an aircraft using the simplex architecture without attack. The safety controller is invoked when $0.025^2-x_2^2< 0$.}
    \label{fig:yaw rate}
\end{figure}

\emph{\underline{Scenario II: The aircraft has no redundant controller.}} In this case, the aircraft restarts the controller to recover the system. We let $N=N_3+N_4$ with $N_3=2$ and $N_4=4$. In this case, we let the adversary attack every $10$ epochs. Algorithm \ref{algo:max range} gives that $c=0$, $d=0.4$, $\tau = 0.18s$ (i.e., 4 epochs), and $u=Kx$ with $K=[0.03017,0.05395,- 4.753\times 10^{-3},7.513\times 10^{-5}]^\top$. The evolution of yaw rate over $150$ epochs is plotted using blue dash-dotted line in Fig. \ref{fig:yaw rate}. The adversary enforces the yaw rate to exceed $0.025$ from the second to eleventh epoch with cost $0.0104 < B$.

\emph{\underline{Scenario III: The aircraft adopts simplex architecture.}} We assume that the main controller is in a faulty condition and produces random control input $u_t\in\mathcal{U}$ for each epoch. The safety controller is invoked once $h(x)< 0$. Once the yaw rate exceeds $0.025$ (from the first to seventh epoch in Fig. \ref{fig:yaw rate}), the safety controller drives the yaw rate to $\mathcal{C}$. Note that the simplex architecture assumes that there exists no adversary, which is different with Scenario I and II.

Therefore, the control policy computed using our proposed algorithm ensures safety of the system with respect to specified budget for any of the CRAs or the simplex architecture.
\section{conclusion}\label{sec:conclusion}

In this paper, we studied the problem of developing a common framework that allows safety analysis and control synthesis of CPS adopting the simplex architecture or the set of cyber resilient architectures including BFT++. We presented the models for cyber and physical subsystems, and formulated the safety property using a budget constraint. Our formulation captures strict safety constraint as a special case. We constructed a hybrid system that models CPS implementing any of these architectures. We derived a set of sufficient conditions for the control policy to satisfy the budget constraint. We translated the conditions into a set of sum-of-squares constraints, and proposed an algorithm to compute the control policy. We analyzed the convergence and complexity of the algorithm. A case study on the lateral control of a Boeing 747 was presented to demonstrate viability of our proposed framework.

\bibliographystyle{IEEEtran}
\bibliography{MyBib}

\end{document}